\documentclass[a4paper,UKenglish]{lipics-v2016}
\pdfoutput=1
\usepackage{microtype}

\usepackage{enumerate}
\usepackage{amssymb}

\usepackage{xcolor,latexsym,amsmath,extarrows,alltt}
\usepackage{xspace}
\usepackage{booktabs}
\usepackage{mathtools}
\usepackage{enumitem}
\usepackage{stmaryrd}
\usepackage{microtype}

\theoremstyle{definition}\newtheorem{algorithm}[theorem]{Algorithm}

\newcommand{\B}{\mathcal{B}}
\newcommand{\C}{\mathcal{C}}
\newcommand{\D}{\mathcal{D}}
\newcommand{\F}{\mathcal{F}}

\newcommand{\V}{\mathcal{V}}
\renewcommand{\O}{\mathcal{O}}

\newcommand{\Rules}{\mathcal{R}}

\newcommand{\Var}{\mathit{Var}}

\newcommand{\Data}{\mathcal{D}\!\!\mathcal{A}}
\newcommand{\Values}{\mathcal{V}\!\!\mathcal{A}}

\newcommand{\arrtype}{\Rightarrow}
\newcommand{\arrz}{\to}
\newcommand{\arrr}{\arrz^*}

\newcommand{\supterm}{\rhd}
\newcommand{\suptermeq}{\unrhd}
\newcommand{\apps}[3]{#1\ #2 \cdots #3}
\newcommand{\arity}{\mathit{arity}}
\newcommand{\typeorder}[1]{\mathit{o}(#1)}

\newcommand{\strue}{\symb{true}}
\newcommand{\sfalse}{\symb{false}}

\newcommand{\exptime}[1]{\ensuremath{\mathtt{EXP}^{#1}\mathtt{TIME}}}
\newcommand{\nexptime}[1]{\ensuremath{\mathtt{NEXP}^{#1}\mathtt{TIME}}}

\renewcommand{\P}{\ensuremath{\mathtt{P}}}
\newcommand{\elementary}{\ensuremath{\mathtt{ELEMENTARY}}}

\newcommand{\symb}[1]{\mathtt{#1}}
\newcommand{\blank}{\textbf{\textvisiblespace}}

\title{Non-deterministic Characterisations%
\footnote{Supported by the Marie Sk{\l}odowska-Curie action ``HORIP'', program H2020-MSCA-IF-2014, 658162.}}

\author[]{Cynthia Kop}
\affil{Department of Computer Science, University of Copenhagen (DIKU)\\
  \texttt{kop@di.ku.dk}}

\authorrunning{C. Kop}

\begin{document}

\maketitle

\begin{abstract}
In this paper, we extend Jones' result---that cons-free programming
with $k^{\text{th}}$-order data and a call-by-value strategy
characterises $\exptime{k}$---to a more general setting, including
pattern-matching and non-deterministic choice.  We show that the
addition of non-determinism is unexpectedly powerful in the
higher-order setting.  Nevertheless, we can obtain a non-deterministic
parallel to Jones' hierarchy result by appropriate restricting rule
formation.
\end{abstract}

\section{Introduction}

In~\cite{jon:01}, Jones introduces \emph{cons-free programming}.
Working with a small functional programming language, cons-free
programs are defined to be \emph{read-only}: recursive data cannot
be created or altered (beyond taking sub-expressions), only read from 
the input.  By imposing further restrictions on data order and
recursion style, classes of cons-free programs turn out to characterise
various deterministic classes in the time and space hierarchies of
computational complexity.
Most relevantly to this work, cons-free programs with data order $k$
characterise the class $\exptime{k}$ of decision problems decidable
in $\O(\exp_2^k(a \cdot n^b))$ on a Turing Machine.

The classes thus characterised are all \emph{deterministic}: they
concern the time and space to solve decision problems on a
deterministic Turing Machine.  As the language considered by Jones is
deterministic, a natural question is whether adding non-deterministic
choice to the language would increase expressivity accordingly.  The
answer, at least in the base case, is \emph{no}: following an early
result by Cook~\cite{coo:73}, Bonfante shows~\cite{bon:06} that adding
a non-deterministic choice operator to cons-free programs with data
order $0$ makes no difference in expressivity: whether with or without
non-deterministic choice, such programs characterise \P.

In this paper, we consider the generalisation of this question: does
adding non-deterministic choice give more expressivity when data of
order greater than $0$ is admitted?  Surprisingly, the answer is yes!
However, we do not obtain the non-deterministic classes; rather,
non-deterministic cons-free programs of any data order $\geq 1$
characterise \elementary, the class $\exptime{0} \cup \exptime{1}
\cup \exptime{2} \cup \dots$.  As this is less useful for complexity
arguments, we amend cons-freeness with a further restriction---unary
variables---which allows us to obtain the expected generalisation:
that (thus restricted) cons-free programs of data order $k$
characterise $\exptime{k}$, whether or not non-deterministic choice
is allowed.

We also generalise Jones' language with pattern matching and
user-defined constructors.

\section{Cons-free programming}

For greater generality---and greater ease of expressing examples---we
extend Jones' language to a limited functional programming language
with pattern matching.  We will use terminology from the term
rewriting world, but very little of the possibilities of this world.

\subsection{Higher-order Programs}

We consider programs using simple types, including product types.
The \emph{type order} $\typeorder{\sigma}$ of a type $\sigma$ is
defined as follows: $\typeorder{\kappa} = 0$ for
$\kappa$ a \emph{sort} (base type), $\typeorder{\sigma \times \tau} =
\max(\typeorder{\sigma},\typeorder{\tau})$ and $\typeorder{\sigma
\arrtype \tau} = \max(\typeorder{\sigma}+1,\typeorder{\tau})$.

Assume given three disjoint set of identifiers: $\C$ of
\emph{constructors}, $\D$ of \emph{defined symbols} and $\V$ of
\emph{variables}; each symbol is equipped with a type.
Following Jones, we limit interest to constructors with a type
$\iota_1 \arrtype \dots \arrtype \iota_m \arrtype \kappa$ where all
$\iota_i$ are types of order $0$ and $\kappa$ is a \emph{sort}.
\emph{Terms} are expressions $s$ such that $s : \sigma$ can be derived
for some type $\sigma$ using the clauses:
\begin{itemize}
\item $\apps{c}{s_1}{s_m} : \kappa$ if $c : \iota_1 \arrtype \dots
  \arrtype \iota_m \arrtype \kappa \in \C$ and each $s_i : \iota_i$
\item $\apps{a}{s_1}{s_n} : \tau$ if $a : \sigma_1 \arrtype \dots
  \arrtype \sigma_n \arrtype \tau \in \V \cup \D$ and each $s_i :
  \sigma_i$
\item $(s,t) : \sigma \times \tau$ if $s : \sigma$ and $t : \tau$
\end{itemize}
Thus, constructors cannot be partially applied, while variables and
defined symbols can be.  If $s : \sigma$, we say $\sigma$ is the type
of $s$, and let $\Var(s)$ be the set of variables occurring in $s$.  A
term $s$ is \emph{ground} if $\Var(s) = \emptyset$.  We say $t$ is a
subterm of $s$, notation $s \suptermeq t$, if either $s = t$ or $s =
\apps{a}{ s_1}{s_n}$ with $a \in \C \cup \F \cup \V$ and $s_i
\suptermeq t$ for some $i$, or $s = (s_1,s_n)$ and $s_i \suptermeq t$
for some $i$.  Note that the head of an application is \emph{not} a
subterm of the application.

A \emph{rule} is a pair of terms $\apps{f}{\ell_1}{\ell_k} \arrz r$
such that (a) $f \in \D$, (b) no defined symbols occur in any
$\ell_i$, (c) no variable occurs more than once in $\apps{f}{\ell_1}{
\ell_k}$, (d) $\Var(r) \subseteq \Var(\apps{f}{\ell_1}{\ell_k})$, and
(e) $r$ has the same type as $\apps{f}{\ell_1}{\ell_k}$.
A \emph{substitution} $\gamma$ is a mapping from variables to ground
terms of the same type, and $s\gamma$ is obtained by replacing
variables $x$ in $s$ by $\gamma(x)$.

We fix a set $\Rules$ of rules, which are \emph{consistent}: if
$\apps{f}{\ell_1}{\ell_k} \arrz r$ and $\apps{f}{q_1}{q_n} \arrz s$
are both in $\Rules$, then $k = n$; we call $k$ the \emph{arity} of
$f$.  The set $\Data$ of \emph{data terms} consists of all ground
constructor terms.  The set $\Values$ of \emph{values} is given by:
(a) all data terms are values, (b) if $v,w$ are values, then $(v,w)$
is a value, (c) if $f \in \D$ has arity $k$, $n < k$ and $s_1,\dots,
s_n$ are values, then $\apps{f}{s_1}{s_n}$ is a value if it is
well-typed.  Note that values whose type is a sort are data terms.
The \emph{call-by-value} reduction relation on ground terms is
defined by:
\begin{itemize}
\item $(s,t) \arrr (v,w)$ if $s \arrr v$ and $t \arrr w$
\item $\apps{a}{s_1}{s_n} \arrr \apps{a}{v_1}{v_n}$ if each $s_i \arrr
  v_i$ and either $a \in \C$, or $a \in \D$ and $n < \arity(a)$
\item $\apps{f}{s_1}{s_m} \arrr w$ if there are values $v_1,\dots,v_m$
  and a rule $\apps{f}{\ell_1}{\ell_n} \arrz r$ with $n \leq m$ and
  substitution $\gamma$ such that each $s_i \arrr v_i = \ell_i\gamma$
  and $\apps{(r\gamma)}{v_{n+1}}{v_m} \arrr w$
\end{itemize}
Note that rule selection is non-deterministic; a choice operator might
for instance be implemented by having rules $\symb{choose}\ x\ y
\arrz x$ and $\symb{choose}\ x\ y \arrz y$.

\subsection{Cons-free Programs}

Since the purpose of this research is to find groups of programs
which can handle \emph{restricted} classes of Turing-computable
problems, we must impose certain limitations. In particular, we will
limit interest to \emph{cons-free} programs:

\begin{definition}\label{def:consfree}
A rule $\ell \arrz r$ is cons-free if for all subterms $r \suptermeq
s$ of the form $s = \apps{c}{s_1}{s_n}$ with $c \in \C$, we have:
$s \in \Data$ or $\ell \supterm s$.
A program is cons-free if all its rules are.
\end{definition}

This definition follows those for cons-free term rewriting
in~\cite{car:sim:14,kop:sim:16} in generalising Jones' definition
in~\cite{jon:01}; the latter fixes the constructors in the program
and therefore simply requires that the only non-constant constructor,
$\symb{cons}$, does not occur in any right-hand side.

In a cons-free program, if $v_1,\dots,v_n,w$ are all data terms, then
any data term occurring in the derivation of $\apps{f}{v_1}{v_n} \arrr
w$ is a subterm of some $v_i$.  This includes the result $w$.

\section{Turing Machines and decision problems}

In this paper, we particularly consider complexity classes of
\emph{decision problems}.  A decision problem is a set $A \subseteq
\{0,1\}^+$.  A deterministic Turing Machine \emph{decides} $A$ in
time $P(n)$ if every evaluation starting with a tape $\blank x_1\dots
x_n\blank\blank\dots$ completes in at most $P(n)$ steps, ending in
the $\mathsf{Accept}$ state if $x_1\dots x_n \in A$ and in the
$\mathsf{Reject}$ state otherwise.

Let $\exp_2^0(m) = m$ and $\exp_2^{k+1}(m) = \exp_2^k(2^{m}) = 2^{
\exp_2^k(m)}$.  The class
$\exptime{k}$ consists of those decision problems which can be decided
in $P(n) \leq \exp_2^k(a \cdot n^b)$ steps for some $a,b$.

\begin{definition}
A program $(\C,\D,\Rules)$ with constructors $\symb{true},
\symb{false} : \symb{bool},\symb{[]} : \symb{list}$ and
$\symb{::}$ (denoted infix) of type $\symb{bool} \arrtype \symb{list}
\arrtype \symb{list}$, and a defined symbol $\symb{start} :
\symb{list} \arrtype \symb{bool}$ \emph{accepts} a decision problem
$A$ if for all $\vec{x} = x_1 \dots x_n \in \{0,1\}^+$: $\vec{x} \in A$
if{f} $\symb{start}\ (\overline{x_1}\symb{::}\dots\symb{::}\overline{
x_n}\symb{::}\symb{[]}) \arrr \symb{true}$, where $\overline{x_i} =
\symb{true}$ if $x_i = 1$ and $\symb{false}$ if $x_i = 0$.
(Note that it is not required that \emph{all} evaluations end in
$\symb{true}$, just that there is at least one---and none if $x \notin
A$).
\end{definition}

\section{A lower bound for expressivity}

To give a lower bound on expressivity, we consider the following
result paraphrased from~\cite{jon:01}:

\begin{lemma}\label{lem:counting}
Suppose that, given an input list $cs ::= \overline{x_1}\symb{::}\dots
\symb{::}\overline{x_n}\symb{::}\symb{[]}$ of length $n$, we have a
representation of $0,\dots,P(n)$, symbols $\symb{seed},\symb{pred},
\symb{zero} \in \D$, and cons-free rules $\Rules$ with:
\begin{itemize}
\item $\symb{seed}\ cs \arrr v$ for $v$ a value representing $P(n)$
\item if $v$ represents $i > 0$, then $\symb{pred}\ cs\ v \arrr w$
  for $w$ a value representing $i-1$
\item if $v$ represents $i$, then $\symb{zero}\ cs\ i \arrr \strue$
  if{f} $i = 0$ and $\symb{zero}\ cs\ i \arrr \sfalse$ if{f} $i \neq
  0$
\end{itemize}
Then any problem which can be decided in time $P(n)$ is accepted by a
cons-free program whose data order is the same as that of $\Rules$,
and which is deterministic if{f} $\Rules$ is.
\end{lemma}

\begin{proof}[Proof Idea]
By simulating an evaluation of a Turing Machine.  This simulation
encodes all transitions of the machine as rules; a transition from
state $i$ to state $j$, reading symbol $r$, writing $w$ and moving
to the right is encoded by a rule $\symb{transition}\ \symb{i}\ 
\symb{r} \arrz (\symb{j},(\symb{w},\symb{R}))$.  In addition, there
are rules for $\symb{state}\ cs\ n$---which returns the state the
machine is in at time $n$---, $\symb{position}\ cs\ n$---which
returns the position of the tape reader---and $\symb{tape}\ cs\ n\ p$%
---for the symbol on the tape at position $p$ and time $n$.  Rules
are for instance: \\
$\phantom{XYZ}\symb{state}\ cs\ n \arrz
  \symb{ifthenelse}\ (\symb{zero}\ cs\ n)\ \symb{Start}\ 
  (\symb{fst}\ (\symb{transitionat}\ cs\ (\symb{pred}\ cs\ n)))$ \\
This returns $\symb{Start}$ at time $0$, and otherwise the state
reduced to in the last transition.
\end{proof}

\begin{example}\label{ex:lincount}
For $P(n) = (n+1)^2-1$, we can represent $i \in \{0,\dots,P(n)\}$ as
any pair $(l_1,l_2)$ of lists, where $i = |l_1| \cdot (n+1) + |l_2|$.
For the counting functions, we define:
\[
\begin{array}{rclcrcl}
\symb{seed}\ cs & \arrz & ([],[]) & &
\symb{zero}\ cs\ (\symb{[]},\symb{[]}) & \arrz & \strue \\
\symb{pred}\ cs\ (xs,y\symb{::}ys) & \arrz & (xs,ys) & &
\symb{zero}\ cs\ (xs,y\symb{::}ys) & \arrz & \sfalse \\
\symb{pred}\ cs\ (x\symb{::}xs,\symb{[]}) & \arrz & (xs,cs) & &
\symb{zero}\ cs\ (x\symb{::}xs,\symb{[]}) & \arrz & \sfalse \\
\end{array}
\]
\end{example}

\begin{lemma}\label{lem:binary}
For any $a,b > 0,k \geq 0$,there are cons-free, deterministic rules
$\Rules_{a,b}^k$ defining counting functions as in
Lemma~\ref{lem:counting} such that, for $P(n) = \exp_2^k(a \cdot
n^b)-1$, the numbers $\{0,\dots,P(n)\}$ can be represented.  All
function variables in $\Rules_{a,b}^k$ have a type $\sigma \arrtype
\symb{bool}$.
\end{lemma}

\begin{proof}[Proof Idea]
For $k = 0$, we can count to $a \cdot n^b-1$ using an approach much
like Example~\ref{ex:lincount}.  Given $\Rules_{a,b}^k$, which
represents numbers as a type $\sigma$, we can define
$\Rules_{a,b}^{k+1}$ by representing a number $i$ with bit vector
$b_0\dots b_M$ (with $M = \exp_2^k(a \cdot n^b)$) as the function in
$\sigma \arrtype \symb{bool}$ which maps a ``number'' $j$ to $\strue$
if $b_i = 1$ and to $\sfalse$ otherwise.
\end{proof}

The observation that the functional variables take only one input
argument will be used in Lemma~\ref{lem:lower} below.
The counting techniques from Example~\ref{ex:lincount} and
Lemma~\ref{lem:binary} originate from Jones' work.  However, in a
non-deterministic system, we can do significantly more:

\begin{lemma}\label{lem:nondetcount}
Let $P_0(n) := n$, and for $k \geq 0$, $P_{k+1}(n) := 2^{P_k(n)}-1$.
Then for each $k$, we can represent all $i \in \{0,\dots,P_k(n)\}$ as
a term of type $\symb{bool}^k \arrtype \symb{list}$, and accompanying
counting functions $\symb{seed}_k,\symb{pred}_k$ and $\symb{zero}_k$
can be defined.
\end{lemma}

\begin{proof}
The base case ($k=0$) is Example~\ref{ex:lincount}.  For larger $k$,
let $i \in \{0,\dots,2^{P_k(n)}-1\}$ have bit vector $b_1\dots b_{P_k(
n)}$; we say $s : \symb{bool}^k \arrtype \symb{list}$ represents $i$
at level $k$ if for all $1 \leq j \leq P_k(n)$: $b_j = 1$ if{f} $s\ 
\strue \arrr v$ for some $v$ which represents $j$ at level $k-1$, and
$b_j = 0$ if{f} $s\ \sfalse \arrr v$ for such $v$.  This relies on
non-determinism: $s\ \strue$ reduces to a representation of
\emph{every} $j$ with $b_j = 1$.  A representation $O$ of $0$ at level
$k-1$ is used as a default, e.g.~$s\ \sfalse \arrr O$ even if
each $b_j = 0$.  The $\symb{zero}$ and $\symb{pred}$ rules rely on
testing bit values, using:

$\symb{bitset}_k\ cs\ F\ j \arrz \symb{bshelp}_k\ cs\ F\ j\ 
  (\symb{equal}_{k-1}\ cs\ (F\ \strue)\ j)\ 
  (\symb{equal}_{k-1}\ cs\ (F\ \sfalse)\ j)$

$\symb{bshelp}_k\ cs\ F\ j\ \strue\ b \arrz \strue$
\phantom{XYZ}
$\symb{bshelp}_k\ cs\ F\ j\ b\ \strue \arrz \sfalse$

$\symb{bshelp}_k\ cs\ F\ j\ \sfalse\ \sfalse \arrz
  \symb{bitset}_k\ cs\ F\ j$. \\
These rules are non-terminating, but if $F$ represents a
number at level $k$, and $j$ at level $k-1$, then $\symb{bitset}_k\ 
cs\ F\ j$ reduces to exactly one value: $\strue$ if $b_j = 1$, and
$\sfalse$ if $b_j = 0$.
\end{proof} 

Thus, we can count up to arbitrarily high numbers; by
Lemma~\ref{lem:counting}, every decision problem in $\elementary$ is
accepted by a non-deterministic cons-free program of data order $1$.

To obtain a more fine-grained characterisation which still admits
non-deterministic choice, we will therefore consider a restriction of
cons-free programming which avoids Lemma~\ref{lem:nondetcount}.

\begin{definition}
A cons-free program \emph{has unary variable} if all variables
occurring in any rule in $\Rules$ have a type $\iota$ or $\sigma
\arrtype \iota$, with $\typeorder{\iota} = 0$.
\end{definition}

Intuitively, in a program with unary variables, functional variables
cannot be \emph{partially applied}; thus, such variables represent a
function mapping to \emph{data}, and not to some complex structure.
Note that the input type $\sigma$ of a unary variable $x : \sigma
\arrtype \iota$ is allowed to be a product $\sigma_1 \times \dots
\times \sigma_n$.
Lemma~\ref{lem:nondetcount} relies on non-unary variables, but
Lemma~\ref{lem:binary} does not.  We obtain:

\begin{lemma}\label{lem:lower}
Any problem in $\exptime{k}$ is accepted by a (non-deterministic)
extended cons-free program of data order $k$.
\end{lemma}

\section{An upper bound for expressivity}

To see that extended cons-free programs \emph{characterise} the
$\exptime{}$ hierarchy, it merely remains to be seen that every
decision problem that is accepted by a call-by-value cons-free program
with unary variables and of data order $k$, can be solved by a
deterministic Turing Machine---or, equivalently, an algorithm in
pseudo code---running in polynomial time.

\begin{algorithm}[Finding the values for given input in a fixed
  extended cons-free program $\Rules$]\label{alg:cbv}
\textbf{\textit{Input:}} a term $\apps{\symb{start}}{v_1}{v_n} :
  \iota$ with each $v_i$ a data term and $\typeorder{\iota} = 0$. \\
\textbf{\textit{Output:}} all data terms $w$ such that
  $\apps{\symb{start}}{v_1}{v_n} \arrr w$.

Let $\B := \bigcup_{1 \leq i \leq m} \{ w \in \Data \mid v_i
  \suptermeq w \} \cup \bigcup_{\ell \arrz r \in \Rules} \{ w \in
  \Data \mid r \suptermeq w \}$.

For all types $\sigma$ occurring as data in $\Rules$, generate
$\llbracket \sigma \rrbracket$ and a relation $\sqsupseteq$, as
follows:
\begin{itemize}
\item $\llbracket \kappa \rrbracket = \{ s \in \B \mid s : \kappa \}$
  if $\kappa$ is a sort;\quad
  for $A,B \in \llbracket \kappa \rrbracket$, let $A \sqsupseteq B$ if
  $A = B$
\item $\llbracket \sigma \times \tau \rrbracket = \{ (A,B) \mid
  A \in \llbracket \sigma \rrbracket \wedge B \in \llbracket \tau
  \rrbracket \}$; \quad
  $(A_1,A_2) \sqsupseteq (B_1,B_2)$ if $A_1 \sqsupseteq B_1$ and
  $A_2 \sqsupseteq B_2$
\item $\llbracket \sigma \arrtype \tau \rrbracket = \mathcal{P}(\{
  (A,B) \mid A \in \llbracket \sigma \rrbracket \wedge B \in
  \llbracket \tau \rrbracket \})$;\quad
  for $A,B \in \llbracket \sigma \arrtype \tau\rrbracket$ let $A
  \sqsupseteq B$ if $A \supseteq B$
\end{itemize}

For all $f : \sigma_1 \arrtype \dots \arrtype \sigma_m \arrtype \iota
\in \D$, note that we can safely assume that $\arity(f) \geq m-1$.
For all such $f$, and all $A_1 \in \llbracket \sigma_1 \rrbracket,
\dots,A_m \in \llbracket \sigma_m \rrbracket, v \in \llbracket \iota
\rrbracket$, note down a statement: $\apps{f}{A_1}{A_m} \approx v$.
If $\arity(f) = m-1$, also note down $\apps{f}{A_1}{A_{m-1}} \approx
O$ for all $O \in \llbracket \sigma_m \arrtype \iota\rrbracket$.

For all rules $\ell \arrz r$, all $s : \tau$ with $r \suptermeq s$
or $s = r\ x$,
all $O \in \llbracket \tau \rrbracket$ and all substitutions $\gamma$
mapping the variables $x : \sigma \in \Var(s)$ to elements of
$\llbracket \sigma \rrbracket$, note down a statement $s\gamma \approx
O$.  Mark all statements $x\gamma \approx O$ such that $x\gamma
\sqsupseteq O$ as \emph{confirmed}, and all other statements
\emph{unconfirmed}.
Repeat the following steps until no new statements are confirmed
anymore.

\begin{itemize}
\item For every unconfirmed statement $\apps{f}{A_1}{A_n} \approx O$,
  determine all rules $\apps{f}{\ell_1}{\ell_k} \arrz r$ (with $k = n$
  or $k = n-1$) and substitutions $\gamma$ mapping $x : \sigma \in
  \Var(\apps{f}{\ell_1}{\ell_k})$ to an element of $\llbracket \sigma
  \rrbracket$, such that each $A_i = \ell_i\gamma$,
  and mark the statement as confirmed if
  $(\apps{r}{x_{k+1}}{x_n})\gamma[x_{k+1}:=A_{k+1},\dots,x_n:=A_n]
  \approx O$ is confirmed.
\item For every unconfirmed statement $(F\ s)\gamma \approx O$,
  mark the statement as confirmed if there exists $A$ with $(A,O) \in
  \gamma(F)$ and $s\gamma \approx A$ is confirmed.
\item For every unconfirmed statement $(\apps{f}{s_1}{s_n})\gamma
  \approx O$, mark it as confirmed if there are $A_1,\dots,A_n$ such
  that both $\apps{f}{A_1}{A_n} \approx O$ and each $s_i\gamma
  \approx A_i$ are confirmed.
\end{itemize}
Then return all $w$ such that $\apps{\symb{start}}{v_1}{v_n} \approx
w$ is marked confirmed.
\end{algorithm}

\begin{lemma}
Algorithm~\ref{alg:cbv} is in $\exptime{k}$---where $k$ is the data
order of $\Rules$---and returns the claimed output.
\end{lemma}

\begin{proof}[Proof Idea]
The complexity of Algorithm~\ref{alg:cbv} is determined by the size of
each $\llbracket \sigma \rrbracket$.  The proof of soundness and
completeness of the algorithm is more intricate; this fundamentally
relies on replacing the values $\apps{f}{v_1}{v_n}$ with $n <
\arity(f)$ by subsets of the set of all tuples $(A,w)$ with the
property that, intuitively, $\apps{f}{v_1}{v_n}\ A \arrr w$.
\end{proof}

\section{Conclusion}

Thus, we obtain the following variation of Jones' result:

\begin{theorem}
A decision problem $A$ is in $\exptime{k}$ if and only if there is a
cons-free program $\Rules$ of data order $k$ and with unary variables,
which accepts $A$.  This statement holds whether or not the program is
allowed to use non-deterministic choice.
\end{theorem}

In addition, we have adapted Jones' language to be more permissive,
admitting additional constructors and pattern matching.  This
makes it easier to specify suitable programs.

Using non-deterministic programs is a step towards further
characterisations; in particular, we intend to characterise
$\nexptime{k} \subseteq \exptime{k+1}$ using restricted
non-deterministic cons-free programs of data order $k+1$.

\bibliography{references}
\bibliographystyle{plain}

\end{document}